\documentclass[a4paper,english,numberwithinsect]{lipics-v2021}

\hideLIPIcs
\nolinenumbers

\usepackage[utf8]{inputenc}
\usepackage{hyperref}
\usepackage{amsmath}
\usepackage{amsfonts}
\usepackage{amssymb}
\usepackage{amsthm}
\usepackage{graphicx}
\usepackage{subcaption}
\usepackage{todonotes}
\usepackage[inline]{enumitem}
\setlist[enumerate,1]{label={(\roman*)}}

\usepackage{xcolor}

\usepackage{xcolor}
\newcommand{\prooflink}[1]{\hypersetup{linkcolor=lipicsYellow}\hyperref[#1]{$\blacktriangledown$}}
\newcommand{\statlink}[1]
{\hypersetup{linkcolor=lipicsBulletGray}\hyperref[#1]{$\blacktriangle$}}

\newtheorem{openquestion}{Open Question}

\graphicspath{ {./figs/} }

\newcommand{\definitionBox}[1]{\smallskip
\noindent\fbox{\begin{minipage}{0.98\textwidth}
\smallskip

#1

\smallskip
\end{minipage}}

\smallskip}

\title{Bowties and Hourglasses: Intersections of Double-Wedges\\\texorpdfstring{\medskip}{}
\Large \textit{Or: Stabbing and Avoiding Line Segments}}

\titlerunning{Bowties and Hourglasses: Intersections of Double-Wedges}
\authorrunning{D. Bertschinger, H. Förster, F. Klute, I. Parada, P. Schnider and B. Vogtenhuber}

\author{Daniel Bertschinger}{ETH Zürich, Zürich, Switzerland}{daniel.bertschinger@inf.ethz.ch}{}{}
\author{Henry F\"orster}{TU Munich, Heilbronn, Germany $\cdot$ University of Tübingen, Tübingen, Germany}{henry.foerster@tum.de}{0000-0002-1441-4189}{partially supported by DFG
grant KA812-18/2}
\author{Fabian Klute}{Universitat Politècnica de Catalunya, Barcelona, Spain}{fabian.klute@upc.edu}{0000-0002-7791-3604}{partially supported by grant PID2023-150725NB-I00 funded by MICIU/AEI/10.13039\-/501100011033. }
\author{Irene Parada}{Universitat Politècnica de Catalunya, Barcelona, Spain}{irene.parada@upc.edu}{0000-0003-3147-0083}{Corresponding author. Serra H\'unter Fellow. Partially supported by grant PID2023-150725NB-I00 funded by MICIU/AEI/10.13039/501100011033}
\author{Patrick Schnider}{ETH Zürich, Zürich, Switzerland $\cdot$ University of Basel, Basel, Switzerland}{patrick.schnider@inf.ethz.ch}{0000-0002-2172-9285}{}
\author{Birgit Vogtenhuber}{TU Graz, Graz, Austria}{bvogt@ist.tugraz.at}{0000-0002-7166-4467}{partially supported by the Austrian Science Fund within the collaborative DACH project \emph{Arrangements and Drawings} as FWF project \mbox{I 3340-N35}}


\acknowledgements{This work was initiated at the 6th DACH Workshop on
Arrangements and Drawings in Stels, Switzerland. We thank the organizers for making this workshop possible, especially in those times, and the other participants for fruitful discussions. We thank the anonymous reviewers for their comments, which helped us strengthen our results.}


\begin{document}

\maketitle

\begin{abstract}
We study the common intersection of arrangements of double-wedges. We consider arrangements where double-wedges may be both bowties (which do not contain a vertical line) or hourglasses (which contain a vertical line), in contrast to earlier studies that focused on arrangements of only bowties. This generalization changes the setting drastically, in particular, with respect to all arguments involving the point-line duality. Namely, a point in the intersection of all double-wedges is equivalent to a line that stabs a set of segments $\mathcal{S}$ (corresponding to the bowties) while it avoids a different set of segments $\mathcal{A}$ (corresponding to the complement of the hourglasses).

We show that in this general setting, the intersection of $n$ double-wedges may consist of $\Omega(n^2)$ interior-disjoint regions. Further, we discuss Gallai-type results for arrangements of segments and anti-segments, and we provide algorithms for computing the intersection of such arrangements with worst-case optimal running time. Finally, we also prove that we can find a single intersection point in almost optimal running time, assuming that \textsc{3SUM} admits no truly subquadratic-time algorithm.

\end{abstract}

\section{Introduction}

Two non-parallel lines $\ell_1$ and $\ell_2$ subdivide the Euclidean plane into four wedges.
The union of two opposite wedges forms a so-called \emph{double-wedge} \cite{CGGMS, HMRS}.
More precisely, it is the closure of the symmetric difference of two half-planes delimited by $\ell_1$ and $\ell_2$. 
We distinguish the two different types of double-wedges, namely, those that do not contain any vertical line, which we call \emph{bowties},
and those that do contain a vertical line, which we call \emph{hourglasses}. 
If neither $\ell_1$ nor $\ell_2$ are vertical, 
they span an hourglass and a bowtie that are the closure of the complement of one another.
For the remainder of this work, we assume without loss of generality that no bounding line of any double-wedge is vertical or horizontal.

The standard \emph{projective duality} $T$ transforms the point $p = (p_x, p_y)$ to the non-vertical line $p^*: y = p_x \cdot x - p_y$ and vice versa. 
Bowties are thus the projective dual of non-vertical line segments in the Euclidean plane;~see Figures~\ref{fig:titlefigure} and~\ref{fig:bowtie}. 
Therefore, each point in the intersection of a family of $n$ {bowties} corresponds~to a stabbing line of the dual family of $n$ line segments. 
This duality is behind the classic algorithm for efficiently stabbing line segments, which effectively computes the intersection of $n$ bowties in time $\mathcal{O}(n \log n)$~\cite{DBLP:journals/bit/EdelsbrunnerMPRWW82}. 
Coming from this line of reasoning, that is, about the stabbing of line segments,
double-wedges are sometimes also defined as what we call bowties~\cite{deBerg, Edelsbrunner}. 

We study more general arrangements of double-wedges, containing both bowties and hourglasses. 
Hourglasses are the projective dual of \emph{anti-segments}, that is, a straight line minus a segment contained in this line; see Figure~\ref{fig:titlefigure} and~\ref{fig:hourglass}.  
Thus, the common intersection of a general double-wedge arrangement corresponds to 
\begin{enumerate*}
\item the stabbing lines of an arrangement of segments and anti-segments, or, equivalently, 
\item the lines that stab all segments corresponding to the bowties while avoiding the segments dual to the complement of each hourglass, see also Figure~\ref{fig:titlefigure}.
\end{enumerate*}

The remainder of this paper is structured as follows. First, we give formal definitions and  preliminary  observations in Section~\ref{sec:preliminaries}. Then, we prove that the intersection of double-wedge arrangements with both bowties and hourglasses may consist of $\Omega(n^2)$ interior-disjoint regions and discuss Gallai-type results for arrangements of segments and anti-segments in Section~\ref{sec:complexity}. In Section~\ref{sec:algorithms} we describe efficient algorithms for the computation of the intersection of a double-wedge arrangement. Finally, in Section~\ref{sec:3sum}, we show that identifying a single intersection point cannot be done much more efficiently than computing the full intersection unless 3SUM admits a truly subquadratic-time algorithm. We conclude  with open problems.

\section{Preliminaries}
\label{sec:preliminaries}
Let $h_1$ and $h_2$ be two half-planes in $\mathbb{R}^2$ that are bounded by non-parallel lines $\ell_1$ and $\ell_2$, respectively. 
We denote the double-wedge formed by the closure of the symmetric difference of $h_1$ and $h_2$ as $\langle h_1, h_2 \rangle$
and the intersection point of $\ell_1$ and $\ell_2$ as the \emph{origin} of $\langle h_1, h_2 \rangle$.

For a double-wedge $d=\langle h_1, h_2 \rangle$, let $\ell$ be the line with slope $a_\ell=(a_1+a_2)/2$ through the origin $o_d$ of $d$, where $a_1$ and $a_2$ are the slopes of the bounding lines $\ell_1$ of $h_1$ and $\ell_2$ of $h_2$, respectively. 
We refer to the two rays of $\ell_1$ and $\ell_2$ emerging from $o_d$ that are located above $\ell$ as the \emph{upper trace} of $d$ and to the other two rays 
of $\ell_1$ and $\ell_2$ as the \emph{lower trace} of $d$; see Figure~\ref{fig:shapes}. 
\begin{figure}[b]
\centering
\includegraphics[scale = 0.68, page =3]{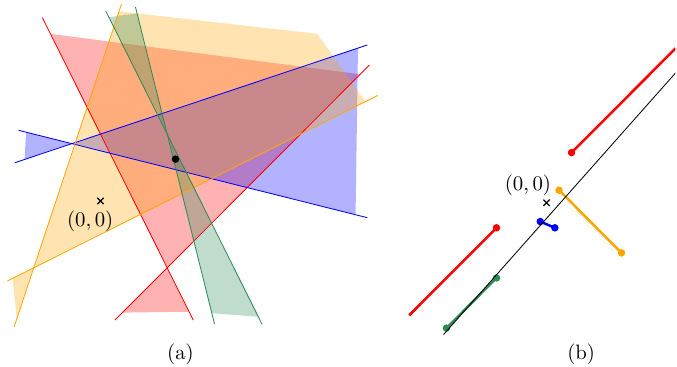}
\caption{(a) An arrangement $\mathcal{D}$ of double-wedges and (b) its projective dual arrangement of segments and anti-segments $\mathcal{A}$. A point contained in all of $\mathcal D$ is dual to a line stabbing all of $\mathcal{A}$.}
\label{fig:titlefigure}
\end{figure}
\begin{figure}[t]
\centering
\begin{subfigure}{0.45\textwidth}
\centering
\includegraphics[scale=0.8,page=1]{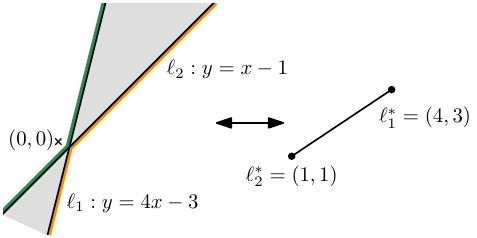}
\centering
\subcaption{}
\label{fig:bowtie}
\end{subfigure}
\hfill
\begin{subfigure}{0.45\textwidth}
\centering
\includegraphics[scale=0.8,page=2]{bowtieAndHourglass}
\centering
\subcaption{}
\label{fig:hourglass}
\end{subfigure}
\caption{(a) A bowtie and its dual line segment. (b)~An hourglass and its dual anti-segment. 
Upper traces are colored green and lower traces orange.}
\label{fig:shapes}
\end{figure}

We say that a point $x \in \mathbb{R}^2$ is \emph{below} the upper trace if and only if a vertical ray emerging from $x$ in positive $y$-direction hits the upper trace. 
Similarly, a point is \emph{above} the lower trace if and only if a vertical ray emerging from $x$ in negative 
$y$-direction hits the lower trace. 
If a point is not below the upper trace (or above the lower trace) we say that the point is above it (or below it, respectively).

\begin{observation}\label{obs:intersectionCharacterization}
Point $p \in \mathbb{R}^2$ is part of the intersection of a set $\mathcal{D}$ of  double-wedges if and only if 
\begin{enumerate*}
\item\label{cond:2} $p$ is below the upper trace but above the lower trace of each bowtie in $\mathcal{D}$, and
\item\label{cond:1} $p$ is above the upper trace or below the lower trace for each hourglass in $\mathcal{D}$.
\end{enumerate*}
\end{observation}

We are also interested in the computational complexity of the following decision problem:

\definitionBox{
\textsc{Double-Wedge Intersection}

\textbf{Input.} An arrangement $\mathcal{D}$ of $n$ double-wedges.

\textbf{Question.} Is there a point $p \in \mathbb{R}^2$ such that $x \in d$ for all $d \in \mathcal{D}$?
}

\textsc{Double-Wedge Intersection} can be solved by computing the full intersection of the double-wedges in $\mathcal{D}$ but a more efficient solution for this decision problem may exist.

By dualizing \textsc{Double-Wedge Intersection} we can define the following equivalent problem:

\definitionBox{
\textsc{Stabbing and Avoiding Segments}

\textbf{Input.} Two sets $\mathcal{S}$ and $\mathcal{A}$ of in total $n$ segments.

\textbf{Question.} Is there a straight line $\ell \subset \mathbb{R}^2$ so that $\forall s \in \mathcal{S}: \ell \cap s \neq \emptyset$  and $\forall a \in \mathcal{A}: \ell \cap a = \emptyset$?
}

Our results 
apply (or can be directly adapted) 
whether 
the endpoints of an (anti-)segment can be stabbed or must be avoided. In the double-wedge arrangement, this corresponds to the cases where the bounding lines are part of the double-wedge or they are not, respectively.

\section{Combinatorial Properties of Double-Wedge Arrangements}
\label{sec:complexity}

\subsection{Combinatorial complexity}

\begin{theorem}\label{thm:complexityOfIntersection}
For every $k \in \mathbb{N}$, there exists a double-wedge arrangement $\mathcal{D}_n$ consisting of $n=2k$ double-wedges so that its intersection consists of $\left(n/2+1 \right)^2$ interior-disjoint regions.
\end{theorem}

\begin{proof}
We first take $k$ bowties whose origins are located on a common horizontal line and whose bounding lines are parallel so that they create a \emph{vertical grating} as shown in Figure~\ref{fig:lowerBound}a (for $k=6)$. The intersection of the vertical grating consists of a $4$-gon between each pair of consecutive (along the $x$-axis) double-wedges plus the two unbounded regions at the left and right boundary; that is, it consists of $k+1$ interior-disjoint regions.

\begin{figure}[t]
\centering
\includegraphics[scale=0.5,page=5]{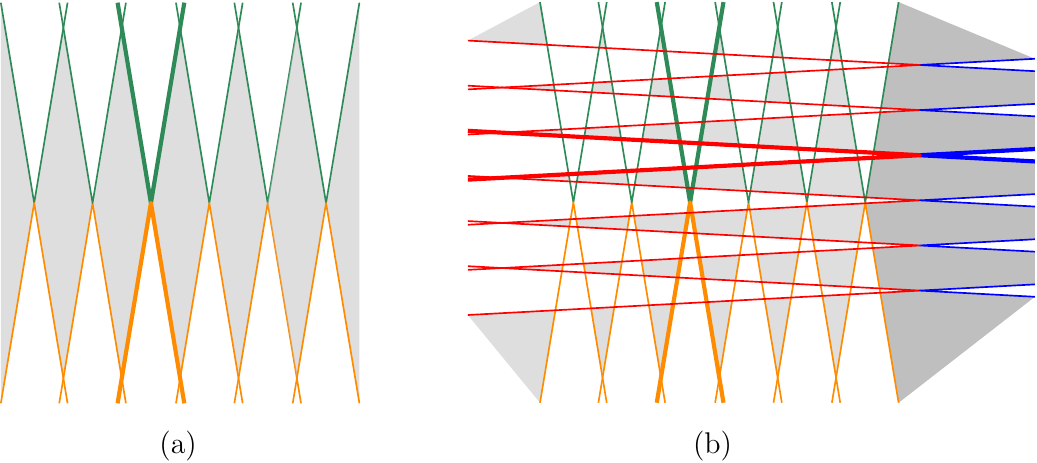}
\caption{Proof of Theorem~\ref{thm:complexityOfIntersection}. The traces of one bowtie in (a) and one bowtie and one hourglass in  (b) are drawn bold. Traces of the horizontal grating are colored red and blue, respectively.}
\label{fig:lowerBound}
\end{figure}

The remaining $k$ double-wedges are a copy of the vertical grating, rotated by $\pi/2$, yielding a \emph{horizontal grating} of hourglasses. 
The entire arrangement $\mathcal{D}_n$ consists of both a vertical and a suitably stretched and translated horizontal grating so that each region of the intersection of the vertical grating intersects each corresponding region of the horizontal one; see Figure~\ref{fig:lowerBound}b (for $k=6$). Thus, the intersection of $\mathcal{D}_n$ 
consists of $(k+1)^2$ interior-disjoint regions.
\end{proof}

Slightly wiggling our construction yields double-wedge arrangements in general position, while a rotation of our construction by $\pi/4$ yields an arrangement that only consists of hourglasses.
In contrast to the intersection of general double-wedge arrangements discussed here,
the intersection of a pure bowtie arrangement  consists only of $\mathcal{O}(n)$ regions~\cite{DBLP:journals/bit/EdelsbrunnerMPRWW82}. 
\medskip

\subsection{Gallai-type results}

A classical result in discrete geometry is \emph{Helly's theorem}. A generalization are \emph{Gallai-type} problems that ask if, given a family $\cal F$ of sets, there exist $k,p\in\mathbb{N}$ such that if any $k$ sets of $\cal F$ have a point in common, then there are $p$ points piercing all the sets in $\cal F$. In the case where the sets are bowties, this problem is well-studied in its dual version: \emph{Given a family of line segments in the plane such that any $k$ of them can be pierced by a line, are there $p$ lines that pierce the entire family?}

In 1935, Vincensini asked whether there is a $k$ for which one can pierce all segments with a single line. A counterexample was found shortly after by Santal\'{o}, showing that no such $k$ can exist~\cite[Chapter 4]{handbook}. On the other hand, Eckhoff showed that for a family of any compact convex sets in the plane, if any four of them can be pierced by a line, then there are two lines piercing the entire family~\cite{Eckhoff1973TransversalenproblemeID}. He also showed that if any three of them can be pierced by a line, then there are four lines that pierce the entire family \cite{DBLP:journals/dcg/Eckhoff93}. He further conjectured that in this case three lines should suffice; it was proven to be the case recently by McGinnis and Zerbib \cite{zerbib}.

In our setting of  segments and anti-segments, we observe that we need at most one more line—a line at infinity that pierces all anti-segments.

\begin{observation}
    Let $\mathcal{F}$ be a finite family of segments and anti-segments in the plane. The following statements are true:
    \begin{enumerate}
        \item If any three elements of $\mathcal{F}$ can be pierced by a line, then $\mathcal{F}$ can be pierced by four lines.
        \item If any four elements of $\mathcal{F}$ can be pierced by a line, then $\mathcal{F}$ can be pierced by three lines.
    \end{enumerate}  
\end{observation}

In contrast, we can show the following lower bound for $p$ given $k=3$:

\begin{theorem}
\label{thm:3lines_necessary}
    There is a family $\mathcal{F}$ of segments and anti-segments in the plane, any three of which can be pierced by a line, where three lines are necessary to pierce $\mathcal{F}$.
\end{theorem}

\begin{proof}
   To explicitly construct such a family, we start by placing nine points in convex position as the vertices of a regular 9-gon. 
   That is, we choose a set of points $P_9=\{p_0,\ldots,p_8\}$ such that for $k\in \{0,\ldots,8\}$ we have \begin{equation*}
        p_k=\left(\cos\left(\frac{2\pi}{9}\cdot k\right),\sin\left(\frac{2\pi}{9}\cdot k\right)\right).
    \end{equation*} 
    We now define one set of segments and two sets of anti-segments using point set $P_9$.
    
    \paragraph*{Green Segments} The first family of segments, which we call the \emph{green segments}, are the segments spanned by two of the points with two points on one and five points on the other side, i.e., $g_k$ is the segment connecting $p_k$ and $p_{(k+3)\bmod{9}}$ for $k\in \{0,...,8\}$. For an illustration, see the green segments in Figure~\ref{fig:lbred}. 
    \begin{figure*}[t]
\centering
\begin{subfigure}{0.24\textwidth}
\centering
\includegraphics[scale=0.6]{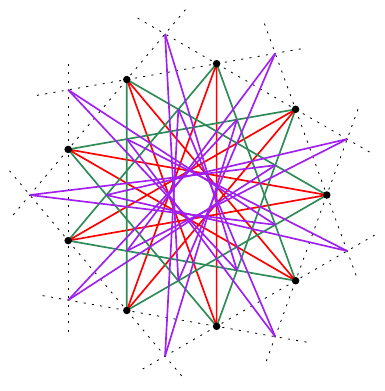}
\subcaption{}
\label{fig:lbcomplete}
\end{subfigure}
\hfill
\begin{subfigure}{0.24\textwidth}
\centering
\includegraphics[scale=0.6]{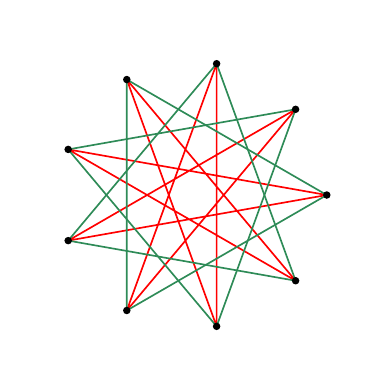}
\subcaption{}
\label{fig:lbred}
\end{subfigure}
\hfill
\begin{subfigure}{0.24\textwidth}
\centering
\includegraphics[scale=0.6]{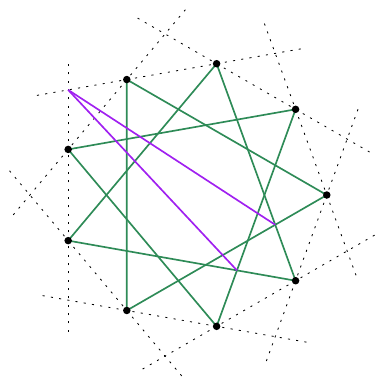}
\subcaption{}
\label{fig:lbpurple}
\end{subfigure}
\hfil
\begin{subfigure}{0.24\textwidth}
\centering
\includegraphics[scale=0.6]{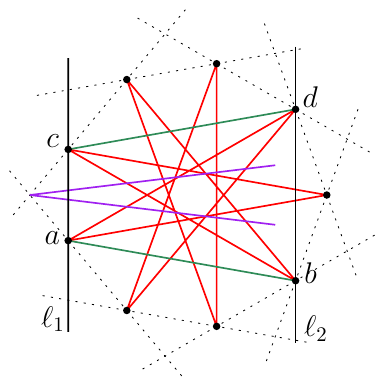}
\subcaption{}
\label{fig:2green}
\end{subfigure}

\begin{subfigure}{0.24\textwidth}
\centering
\includegraphics[scale=0.6]{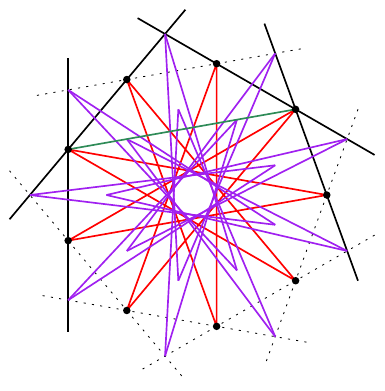}
\subcaption{}
\label{fig:1green}
\end{subfigure}
\hfil
\begin{subfigure}{0.24\textwidth}
\centering
\includegraphics[scale=0.6]{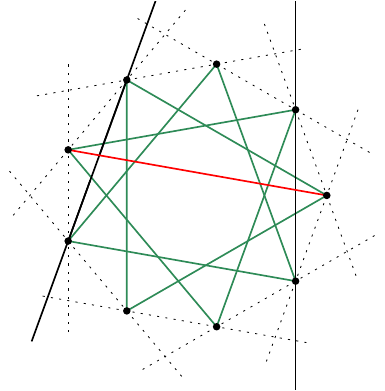}
\subcaption{}
\label{fig:NoPiercingCase1}
\end{subfigure}
\hfil
\begin{subfigure}{0.24\textwidth}
\centering
\includegraphics[scale=0.6]{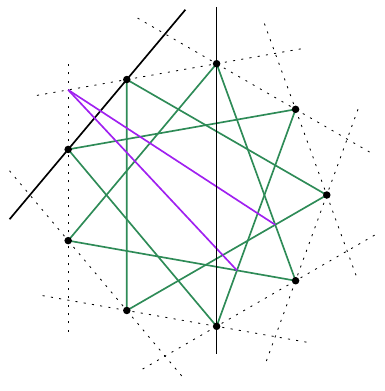}
\subcaption{}
\label{fig:NoPiercingCase2}
\end{subfigure}

\caption{Proof of Theorem~\ref{thm:3lines_necessary} (for anti-segments, the complement segments are shown): (a)~$\mathcal{F}$. (b)~Green segments and red anti-segments. (c)~Two purple anti-segments. (d)~Two parallel green segments and one anti-segment can always be pierced.  (e)~One green segment and two anti-segments can always be pierced. (f)--(g)~Non-pierceability by two lines. }
\end{figure*}
\paragraph*{Red Anti-Segments} We then add the first set of anti-segments, which we call the \emph{red anti-segments}. We define these anti-segments through their complements, which we depict as (open) red segments. These segments are again spanned by points of the point set, this time with three points on one side and four points on the other side. In other words, for $k\in \{0,...,8\}$, anti-segment $r_k$ is the anti-segment that avoids the segment between points $p_k$ and $p_{(k+4)\bmod 9}$;  see Figure~\ref{fig:lbred} for an illustration. Note that a point of the point set intersects two of the green segments but none of the red segments, as they are open.

\paragraph*{Purple Anti-Segments} We now add a second set of anti-segments, called the \emph{purple anti-segments}, which we again define through their complements, which we depict by purple segments. For ease of construction, we consider these purple segments to be closed, but we can make them open by prolonging them by $\varepsilon$ in both directions and removing the endpoints. The purple segments are spanned by two endpoints. 
More precisely, for each $k\in \{0,...,8\}$, two purple anti-segments $q_k$ and $\tilde{q}_k$ share an endpoint $w_k$, which is the intersection of the two supporting lines of the segments $\overline{p_{k-1}p_{k}}$ and $\overline{p_{k+1}p_{k+2}}$ on the nonagon $[p_0,\ldots,p_k]$. The second endpoint of $q_k$ is the intersection of the two green segments $g_{k+2}$ and $g_{k+4}$ whereas the other endpoint of $\tilde{q}_k$ is the intersection of the two  green segments $g_{k+3}$ and $g_{k+5}$; see Figure~\ref{fig:lbpurple} for an illustration. 

\noindent The full construction can be found in Figure~\ref{fig:lbcomplete}.

\begin{claim}\label{lem:3piercing}
Let $\mathcal{F}$ be the family of segments and anti-seg\-ments defined above. Then any three elements of $\mathcal{F}$ can be pierced by a line.
\end{claim}

\begin{claimproof}
We make a case distinction on the number of green segments. The cases of three or zero green segments are easy to see. If there are two green segments, we distinguish two subcases, depending on whether the two green segments intersect or not.

Assume first that the two green segments intersect, and let $p$ be their point of intersection. Further, let $s$ be the red or purple segment that is the complement of the anti-segment in the set we are analyzing. We want to show that there is a line $\ell$ which intersects the two green segments but does not intersect $s$. This can be achieved by drawing $\ell$ through $p$ almost parallel to $s$.

Assume now that the two green segments $g_a=\overline{ab}$ and $g_c=\overline{cd}$ do not intersect. Then, up to relabeling, we must have that $a$ and $c$ are consecutive on the convex hull, whereas there is one vertex between $b$ and $d$; see Figure~\ref{fig:2green}. Let $\ell_1$ be the line through $a$ and $c$. This line misses all red segments and all but two purple segments. The line $\ell_2$ through $b$ and $d$ misses these two purple segments. Thus either $\ell_1$ or $\ell_2$ intersects $g_a$ and $g_c$ but misses the complement of the anti-segment.

This leaves us with the case of one green segment. Up to renaming, we assume that the green segment is $g_1$. Consider the four supporting lines $\ell_0$, $\ell_1$, $\ell_3$ and  $\ell_4$ of the segments $\overline{p_0p_1}$, $\overline{p_1p_2}$, $\overline{p_3p_4}$ and $\overline{p_4p_5}$, respectively; see Figure~\ref{fig:1green}.
By slightly rotating lines $\ell_0$ and $\ell_1$ clockwise and $\ell_3$ and $\ell_4$ counter-clockwise  around their  points shared with $g_1$, we avoid parallelism to anti-segments and ensure that \begin{itemize}
    \item $\ell_0$ pierces all the red anti-segments and all purple anti-segments except $q_0$, $\tilde{q}_0$, $q_1$ and $\tilde{q}_1$,
    \item $\ell_1$ pierces all the red anti-segments and all purple anti-segments except $q_1$, $\tilde{q}_1$, $q_2$ and $\tilde{q}_2$,
    \item $\ell_3$ pierces all the red anti-segments and all purple anti-segments except $q_2$, $\tilde{q}_2$, $q_3$ and $\tilde{q}_3$,
    \item $\ell_4$ pierces all the red anti-segments and all purple anti-segments except $q_3$, $\tilde{q}_3$, $q_4$ and $\tilde{q}_4$.
\end{itemize}

Based on this analysis of non-pierced anti-segments, it can now be seen that for any pair of red or purple anti-segments, one of these four lines pierces both, i.e., it misses the two complements of these two anti-segments. 
\end{claimproof}

\begin{claim}\label{lem:2missing}
Let $\mathcal{F}$ be the family of segments and anti-seg\-ments defined above. Then at least two lines are necessary to pierce all green segments. Further, any two lines that pierce all green segments miss some red or purple anti-segment.
\end{claim}

\begin{claimproof}
It is straightforward to check that no single line can pierce all green segments: 
since the green segments form three triangles, such a line would need to pass through three non-collinear vertices, which is impossible.

Consider now two lines that pierce all green segments. We distinguish two cases, depending on whether both of them intersect the interior of the convex hull of the 9 points or not. If both of them intersect the convex hull, then by construction of the red segments there is a red segment pierced by both lines; see Figure~\ref{fig:NoPiercingCase1}. This corresponds to an anti-segment that is not pierced by the two lines. If neither line intersects the interior of the convex hull, then one of them needs to be tangent to it. Observe that the tangent line can pierce at most four green segments, i.e.,  the other line must pierce at least five green segments, and thus must pass further inside the convex hull. In this case, by construction of the purple segments, there is a purple segment hit by both lines; see Figure~\ref{fig:NoPiercingCase2}.
\end{claimproof}

By combining Claims~\ref{lem:3piercing} and \ref{lem:2missing}, we get Theorem~\ref{thm:3lines_necessary}.
\end{proof}

\section{Computing the Intersection of Double-Wedges}
\label{sec:algorithms}
In this section, we present algorithms for computing the (common) intersection of an arrangement $\mathcal{D}$ of double-wedges in general and in particular cases.   

\subsection{The double-wedges in \texorpdfstring{$\mathcal{D}$}{D} do not cover all slopes} 
\label{sec:algoNoHourglass}
An $\mathcal{O}(n \log n)$-time algorithm~\cite{DBLP:journals/bit/EdelsbrunnerMPRWW82} is known for the case where all double-wedges are bowties. This algorithm can also be applied in the case where the set $\mathcal{D}$ of double-wedges does not cover all slopes in the plane, where $\mathcal{D}$ covers a slope $a \in \mathbb{R}\cup\{\infty\}$ if there exists a double-wedge in $\mathcal{D}$ containing a line of slope $a$.
In order to achieve this, we identify a slope $a$ not covered by any double-wedge and rotate $\mathcal{D}$ such that lines with slope $a$ become vertical. 
It is noteworthy that if $a$ exists, we may assume without loss of generality that $a$ is rational since the rationals are dense in the real line. This procedure results in a pure bowtie arrangement.

\begin{theorem}
\label{thm:algoSpecialCase}
Let $\mathcal{D}$ be an arrangement of double-wedges so that $\mathcal{D}$ does not cover all slopes in $\mathbb{R} \cup \{\infty=-\infty\}$. Then, the intersection of $\mathcal{D}$ can be computed in time $\mathcal{O}(n \log n)$.
\end{theorem}

\begin{proof}
   It remains to argue that $a$ can be computed in time $\mathcal{O}(n \log n)$. 
We first sort the bounding lines of all double-wedges by slope in time $\mathcal{O}(n \log n)$ and check for slope $a^*=-\infty$ in how many double-wedges it lies. In fact, this number is equal to the number of hourglasses and it can be found in linear time by checking for every double-wedge individually whether it is an hourglass or not. If it is zero, then we can directly apply the algorithm from~\cite{DBLP:journals/bit/EdelsbrunnerMPRWW82}. Otherwise, we iteratively increase $a^*$ and whenever $a^*$ becomes larger than the slope of a line bounding a double-wedge $d$, we update the number of double-wedges $a^*$ is in. 
Note that the number of double-wedges covering a slope only changes by $\pm1$ on each such event (and does not change between two neighboring slopes in the sorted slope list). 
Hence, going once through the sorted list of slopes, each update can be done in constant time (by checking the double-wedge $d$).
This procedure takes $\mathcal{O}(n)$ time
and identifies two slopes of bounding lines $a_1$ and $a_2$ so that no slope between $a_1$ and $a_2$ is covered by any double-wedge in $\mathcal{D}$. 
We choose any slope between $a_1$ and $a_2$ for $a$, e.g., $a=(a_1+a_2)/2$.   
\end{proof}

Edelsbrunner et al.~\cite{DBLP:journals/bit/EdelsbrunnerMPRWW82} also state that the computation of the intersection of $n$ bowties requires  $\Omega(n \log n)$ time. 
Since this is a special case of the instances discussed here, the running time of our algorithm is worst-case optimal. We further emphasize that the algorithm for identifying $a$ can also be used to check whether $\mathcal{D}$ covers all slopes in time $\mathcal{O}(n \log n)$. 

\subsection{General setting} 

\label{sec:algoHourglass}
For the general setting, and in particular if the double-wedges in {$\mathcal{D}$} cover all slopes, we can compute in time $\mathcal{O}(n^2)$ the line arrangement formed by the lines bounding the $n$ double-wedges in $\mathcal{D}$. 
To store it, we use the doubly-connected edge list (DCEL) data structure ~\cite{MullerP78}. 
In it, each edge of the line arrangement is decomposed into two directed half-edges pointing in opposite directions. 
(To handle unbounded edges one can add an infinite vertex connected to all of them or place a large bounding box.)
The DCEL contains a record for each half-edge, vertex, and cell.
Each half-edge stores pointers to its target vertex, its twin half-edge, the cell it bounds, and the next and previous half-edges in the cell's boundary cycle. 
Vertices store a pointer to one of their incident half-edges, while cells store a pointer to one of their boundary half-edges.

It is a well-known consequence of the Zone theorem that we can construct the DCEL incrementally in $\mathcal{O}(n^2)$ time, inserting the $2n$ lines bounding the $n$ double-wedges one at a time~\cite{deBerg}.
Now the idea is to explore all the cells in the line arrangement traversing its planar-dual graph in BFS order. 
The planar-dual $G^*$ of a planar graph such as the planarization $G$ of our line arrangement has an edge for each pair of cells in $G$ that are separated by an edge. 
For our algorithm we do not need to compute $G^*$ explicitly; the DCEL suffices. 

Choose an arbitrary cell $f_0$ of the arrangement and compute the number of double-wedges of $\mathcal{D}$ in which $f_0$ is contained in $\mathcal{O}(n)$ time.
Then, traverse the cells of the arrangement starting from~$f_0$ in BFS order in $G^*$-- the DCEL allows accessing the neighboring cells of a cell $f$ in constant time per neighboring cell of $f$. 
Moving to a neighboring cell corresponds to either leaving or entering a single double-wedge of $\mathcal{D}$ as characterized in Observation~\ref{obs:intersectionCharacterization}. 
Whenever we visit a new cell in the arrangement during the traversal, we can determine the number of double-wedges in $\mathcal{D}$ that contain this cell in constant time. 
We can thus determine all cells in the intersection of $\mathcal{D}$ in $\mathcal{O}(n^2)$ time.

\begin{theorem}\label{thm:algoGeneralCase}
Let $\mathcal{D}$ be an arrangement of $n$ double-wedges. Then the intersection of $\mathcal{D}$ can be computed in time $\mathcal{O}(n^2)$.
\end{theorem}

By Theorem~\ref{thm:complexityOfIntersection}, the intersection of $\mathcal{D}$ can consist of $\Omega(n^2)$ interior-disjoint regions. Thus, enumerating all those regions takes time $\Omega(n^2)$ and our algorithm is worst-case optimal. Our approach can also be used to identify the set of cells contained in at least $b$ bowties and $h$ hourglasses.

\subsection{Only few hourglasses in \texorpdfstring{$\mathcal{D}$}{D}}

We can obtain the following result, which is an improvement for the case in which there are only $k \ll n$ hourglasses in $\mathcal{D}$. 
In this case, the intersection can be computed in $\mathcal{O}(kn\log n)$ time.

\begin{theorem}
\label{thm:algoParameterized}
    Let $\mathcal{D}$ be an arrangement of $n$ double-wedges, $k$ of which are hourglasses. Then the intersection of $\mathcal{D}$ can be computed in time $\mathcal{O}(k^2 + nk)\log (n + k^2))$.
\end{theorem}

\begin{proof}
    We begin by computing the intersection $\mathcal{I}_b$ of the $n-k$ bowties using the algorithm by Edelsbrunner et al.~\cite{DBLP:journals/bit/EdelsbrunnerMPRWW82} in time $\mathcal{O}(n \log n)$. Moreover, we compute the intersection $\mathcal{I}_h$ of the $k$ hourglasses using Theorem~\ref{thm:algoGeneralCase} in time $\mathcal{O}(k^2)$. According to Edelsbrunner et al.~\cite{DBLP:journals/bit/EdelsbrunnerMPRWW82}, $\mathcal{I}_b$ is a polygonal domain consisting of $n+1$ convex polygonal regions with $\mathcal{O}(n)$ vertices and edges in total, whereas $\mathcal{I}_h$ is a polygonal domain with $\mathcal{O}(k^2)$ vertices and edges, because it is defined by $\mathcal{O}(k)$ lines. Thus, according to \cite[Theorem 2.6]{deBerg}, the intersection of $\mathcal{I}_b$ and $\mathcal{I}_h$ can be computed in time $\mathcal{O}((n+k^2)\log(n+k^2)+c\log(n+k^2))$ where $c$ denotes the size of the intersection of $\mathcal{I}_b$ and $\mathcal{I}_h$. 
    For $c$, observe that it is defined by the sizes of $\mathcal{I}_b$ and $\mathcal{I}_h$ and the number of intersections between them. 
    The boundary of each of the $n+1$ convex polygonal regions in $\mathcal{I}_b$ can have at most two intersection points with each of the $2k$ lines bounding the $k$ hourglasses. 
    In the worst case, we have $\mathcal{O}(nk)$ new boundary pieces for the intersection between $\mathcal{I}_b$ and $\mathcal{I}_h$.
    Hence, $c=\mathcal{O}(k^2 + nk)$ and the theorem follows.
    \end{proof}

\section{3SUM-Hardness of Stabbing and Avoiding Segments}
\label{sec:3sum}

\noindent Note that Theorem~\ref{thm:algoGeneralCase} immediately implies the following:

\begin{corollary}\label{cor:complexity}
There are $\mathcal{O}(n^2)$-time algorithms for \textsc{Double-Wedge Intersection} and \textsc{Stabbing and Avoiding Segments}.
\end{corollary}

The result by Edelsbrunner et al.~\cite{DBLP:journals/bit/EdelsbrunnerMPRWW82} implies the existence of $\mathcal{O}(n \log n)$-time algorithms in the special cases where $\mathcal{D}$ contains no hourglasses in \textsc{Double-Wedge Intersection} and where $\mathcal{A}=\emptyset$ in \textsc{Stabbing and Avoiding Segments}. We provide evidence that our $\mathcal{O}(n^2)$-time algorithms are likely almost worst-case optimal by reduction from the following problem.

\definitionBox{
\textsc{3SUM}

\textbf{Input.} A set $A \subset \mathbb{R}$ of $n$ real numbers.

\textbf{Question.} Are there three numbers $a,b,c \in A$ such that $a+b+c=0$?
}

There are $o(n^2)$-time algorithms for \textsc{3SUM}~\cite{DBLP:journals/talg/Chan20,DBLP:journals/jacm/GronlundP18}, but the following is widely believed:

\begin{conjecture}[No Truly Subquadratic 3SUM~\cite{DBLP:conf/focs/AbboudW14,DBLP:journals/comgeo/GajentaanO95,DBLP:conf/stoc/Patrascu10}]\label{con:3sum} There is no $\mathcal{O}(n^{2-\varepsilon})$-time algorithm for \textsc{3SUM} for any $\varepsilon>0$.
\end{conjecture}

The importance of Conjecture~\ref{con:3sum} lies in the fact that many (geometric) problems have been classified as \emph{\textsc{3SUM}-hard}~\cite{DBLP:journals/comgeo/GajentaanO95}, i.e., for many problems $P$ there is a subquadratic-time and linear-space reduction from \textsc{3SUM} to $P$. Such a reduction implies that $P$ cannot be solved in truly subquadratic time unless Conjecture~\ref{con:3sum} fails. 
Gajentaan and Overmars~\cite{DBLP:journals/comgeo/GajentaanO95} show \textsc{3SUM}-hardness for several geometric problems. Crucially, they establish \textsc{3SUM}-hardness for the following geometric problem—a fact we use in the proof of Theorem~\ref{thm:3sum}; see also Figure~\ref{fig:3sum:1}.

\definitionBox{
\textsc{GeomBase}

\textbf{Input.} A set $U \subset \mathbb{N}^2$ of $n$ points on three horizontal lines $A\colon y=0$, $B\colon y=1$, $C\colon y=2$.

\textbf{Question.} Is there a non-horizontal line $\ell \subset \mathbb{R}^2$ containing three points of $U$?
}

\begin{theorem}[\cite{DBLP:journals/comgeo/GajentaanO95}]
\textsc{GeomBase} cannot be solved in  $\mathcal{O}(n^{2-\varepsilon})$ time for any $\varepsilon>0$ unless Conjecture~\ref{con:3sum} fails.
\end{theorem}

\begin{theorem}
\label{thm:3sum}
\textsc{Stabbing and Avoiding Segments} cannot be solved in  $\mathcal{O}(n^{2-\varepsilon})$ time for any $\varepsilon>0$ unless Conjecture~\ref{con:3sum} fails.
\end{theorem}

\begin{proof}
We provide a linear time and space reduction from \textsc{GeomBase} to \textsc{Stabbing and Avoiding Segments} as in the reduction by Gajentaan and Overmars~\cite{DBLP:journals/comgeo/GajentaanO95} for a problem they call \textsc{Separator2}.

\begin{figure}[t!]
\centering
\begin{subfigure}{0.45\textwidth}
\centering
\includegraphics[scale=1, page =1]{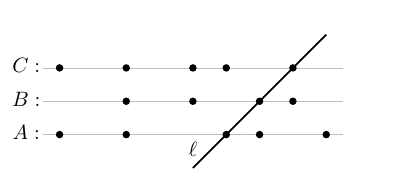}
\centering
\subcaption{}
\label{fig:3sum:1}
\end{subfigure}
\hfill
\begin{subfigure}{0.45\textwidth}
\centering
\includegraphics[scale=1,page=2]{3sum}
\centering
\subcaption{}
\label{fig:3sum:2}
\end{subfigure}
\caption{(a) Instance of \textsc{GeomBase} with  valid solution $\ell$. (b) Reduction to   \textsc{Stabbing and Avoiding Segments} -- the blue segment $s^*$ is the only segment in $\mathcal{S}$ and red segments belong to $\mathcal{A}$.}
\label{fig:3sum}
\end{figure}
Let 	$U$ denote the input of the \textsc{GeomBase} instance; see Figure~\ref{fig:3sum:1}. Assume without loss of generality that the minimum $x$-coordinate in $U$ is $1$ and that the maximum $x$-coordinate in $U$ is $X$. We now construct an instance of \textsc{Stabbing and Avoiding Segments} as follows; see Figure~\ref{fig:3sum:2}.

To construct $\mathcal{A}$, create a rectangle $R$ defined by points $(0,0)$, $(0,2)$, $(X+1,2)$ and $(X+1,0)$. Initially, $\mathcal{A}$ consists of four segments bounding $R$ and the segment between $(0,1)$ and $(X+1,1)$. Now, for each point $p \in U$ remove the part of the segment $a \in \mathcal{A}$ traversing $p$ in a  disk of radius $\varepsilon=\frac{1}{4}$ centered at $p$, subdividing $a$ in the process. Finally, $\mathcal{S}$ consists of the single segment $s^*$ between points $(\varepsilon,\frac{1}{2})$ and $(X+1-\varepsilon,\frac{1}{2})$. This concludes our construction. It now remains to show that the instance $U$ of \textsc{GeomBase} is positive if and only if $\langle \mathcal{S}, \mathcal{A}\rangle$ is a positive instance of  \textsc{Stabbing and Avoiding Segments}.

First, assume that there is a non-horizontal line $\ell$ intersecting three points of $U$.  Clearly, line $\ell$ must traverse the rectangle $R$ from the top edge to the bottom edge. Thus, it also intersects $s^*$. Further, it avoids all segments in $\mathcal{A}$ as we ensured that an $\varepsilon$-disk surrounding each point $p\in U$ is free of segments in $\mathcal{A}$. Thus, line $\ell$ is also a valid solution for our  constructed \textsc{Stabbing and Avoiding Segments} instance. 

Second, assume that there is a line $\ell$ stabbing $s^*$ and avoiding the lines in $\mathcal{A}$. Since $\mathcal{A}$ contains the vertical boundaries of $R$ between points $(0,0)$ and $(0,2)$ as well as between points $(X+1,0)$ and $(X+1,2)$ and $s^*$ is located wholly inside $R$, line $\ell$ must be non-horizontal and intersect the lines $A\colon  y= 0$, $B\colon y=1$ and $C\colon y=2$ in the $x$-range $(0,X+1)$. Intuitively, we can make now use of the  observation that such a line $\ell$ must have passed through three disks with radius $\varepsilon$ and centers on lines $A$, $B$ and $C$. The choice of $\varepsilon = 1/4$ ensures that the three centers must be roughly collinear. Therefore, we obtain a stabbing line corresponding roughly to a solution to \textsc{GeomBase}.

More precisely, recall that by construction, each  segment of $A$, $B$ and $C$ that can be traversed in the $x$-range $(0,X+1)$ is centered on an integer grid point and has length $2\varepsilon$. Hence, line $\ell$ intersects $a$, $b$ and $c$ at $x$-coordinates $(a+\alpha)$, $(b+\beta)$ and $(c+\gamma)$, respectively, where $a,b,c \in \mathbb{N}$ and $\alpha,\beta,\gamma \in (-\varepsilon,\varepsilon)$. Moreover, note, that $(a,0), (b,1), (c,2) \in U$. Since $\ell$ is a straight line, we have:
\begin{equation*} \label{eq1}
\begin{split}
(c+\gamma)-(a+\alpha) &= 2((b+\beta)-(a+\alpha)) \\
(c+\gamma) &=2(b+\beta)-(a+\alpha)\\
2(b+\beta) &=(c+\gamma)+(a+\alpha).
\end{split}
\end{equation*}

As we have chosen $\varepsilon=\frac{1}{4}$, we also have $|\alpha+\gamma|,|2\beta|<\frac{1}{2}$, and thus, we must have $a+c=2b$ and thus $c = 2b-a$. 

Now define $\ell':y=\frac{1}{b-a} \cdot x - \frac{a}{b-a}$. Indeed, $\ell'$ traverses all of $(a,0), (b,1), (c,2) \in U$:
\begin{eqnarray*}
\frac{1}{b-a} \cdot a - \frac{a}{b-a} &=& 0 \quad \checkmark \\
\frac{1}{b-a} \cdot b - \frac{a}{b-a} &=& \frac{b-a}{b-a} = 1 \quad \checkmark\\
\frac{1}{b-a} \cdot c - \frac{a}{b-a} &=& 
\frac{2b-2a}{b-a} = 2 \quad \checkmark
\end{eqnarray*}
Thus, line $\ell'$ is a valid solution for the \textsc{GeomBase} problem. This concludes the proof.
\end{proof}

Due to the equivalence to \textsc{Stabbing and Avoiding Segments}, we obtain:

\begin{corollary}\label{col:3sum}
\textsc{Double-Wedge Intersection} cannot be solved in $\mathcal{O}(n^{2-\varepsilon})$ time for any $\varepsilon>0$ unless Conjecture~\ref{con:3sum} fails.
\end{corollary}

Finally, we remark that
    Corollary~\ref{col:3sum} holds even if $\mathcal{D}$ contains only a single bowtie, in contrast to Theorem~\ref{thm:algoParameterized}, which asserts a better run time for the case in which there are few hourglasses.

\section{Open problems}
\label{sec:conclusions}
 
In this paper, we presented efficient and worst-case optimal algorithms for computing the intersection of a double-wedge arrangement, which is the dual formulation of the problem of finding a line stabbing certain prespecified line segments that keeps another set of line segments untouched. 
An intriguing open problem would be to investigate generalizations of double-wedge arrangements in higher dimensions and generalize our results to this setting. 
That is, one can study the complexity of the resulting arrangements and worst-case time-optimal algorithms.

In addition, we proved that there are arrangements of segments and anti-segments, any three of which can be pierced by a line, such that three lines are necessary to pierce all of them. 
With respect to this result, we ask the same question for a family of line segments:
\begin{openquestion}\label{conj:2lines}
    Let $\mathcal{F}$ be a finite family of line segments in the plane such that any three of them can be pierced by a line. Are there always two lines that pierce~$\mathcal{F}$?
\end{openquestion}

We remark that a positive answer to this open question would imply that a family of segments and anti-segments, any three of which can be pierced by a line,  could be pierced by three lines—making Theorem~\ref{thm:3lines_necessary} tight.

\small
\bibliographystyle{abbrv}
\bibliography{references}

\end{document}